\documentclass[a4paper,11pt]{article}

\usepackage{fullpage}
\usepackage{times}

\usepackage{amssymb,amsmath,amsfonts,amstext,amsthm}
\usepackage{url}

\usepackage{multirow}

\newtheorem{definition}{Definition}[section]
\newtheorem{theorem}{Theorem}[section]
\newtheorem{lemma}[theorem]{Lemma}

\newtheorem{corollary}[theorem]{Corollary}

\newcommand{\ceil}[1]{ \lceil #1 \rceil }
\newcommand{\floor}[1]{ \lfloor #1 \rfloor }

\title{{\bf Near-optimal asymmetric binary matrix partitions}\thanks{This work was partially supported by the European Social Fund and Greek national funds through the research funding program Thales on ``Algorithmic Game Theory'' and by the Caratheodory research grant E.114 from the University of Patras.}}

\author{Fidaa Abed\thanks{Max-Planck-Institut f\"ur Informatik, Saarbr\"ucken, Germany. Email: {\tt fabed@mpi-inf.mpg.de}} \and
Ioannis Caragiannis\thanks{Computer Technology Institute and Press ``Diophantus'' \& Department of Computer Engineering and Informatics, University of Patras, 26504 Rion, Greece. Email: {\tt caragian@ceid.upatras.gr}}  \and  Alexandros A. Voudouris\thanks{Department of Computer Engineering and Informatics, University of Patras, 26504 Rion, Greece. Email: {\tt voudouris@ceid.upatras.gr}}}

\date{}

\begin{document}

\maketitle

\allowdisplaybreaks

\begin{abstract}
We study the asymmetric binary matrix partition problem that was recently introduced by Alon et al.~(WINE 2013) to model the impact of asymmetric information on the revenue of the seller in take-it-or-leave-it sales. Instances of the problem consist of an $n \times m$ binary matrix $A$ and a probability distribution over its columns. A {\em partition scheme} $B=(B_1,...,B_n)$ consists of a partition $B_i$ for each row $i$ of $A$.
The partition $B_i$ acts as a smoothing operator on row $i$ that distributes the expected value of each partition subset proportionally to all its entries.
Given a scheme $B$ that induces a smooth matrix $A^B$, the partition value is the expected maximum column entry of $A^B$.
The objective is to find a partition scheme such that the resulting partition value is maximized.
We present a $9/10$-approximation algorithm for the case where the probability distribution is uniform and a $(1-1/e)$-approximation algorithm for non-uniform distributions, significantly improving results of Alon et al. 
Although our first algorithm is combinatorial (and very simple), the analysis is based on linear programming and duality arguments.
In our second result we exploit a nice relation of the problem to submodular welfare maximization.
\end{abstract}

\section{Introduction} \label{sec:intro}

We study the \textit{asymmetric binary matrix partition problem}, recently proposed by Alon et al.~\cite{AFGT13}. 
Consider a matrix $A \in \{0,1\}^{n \times m}$ and a probability distribution $p$ over its columns; $p_j$ denotes the probability associated with column $j$. We distinguish between two cases for the probability distribution over the columns of the given matrix, depending on whether it is uniform or non-uniform. A partition scheme $B=(B_1,...,B_n)$ for matrix $A$ consists of a partition $B_i$ of $[m]$ for each row $i$ of $A$. More specifically, $B_i$ is a collection of $k_i$ pairwise disjoint subsets $B_{ik} \subseteq [m]$ (with $1 \leq k \leq k_i$) such that $\bigcup_{k=1}^{k_i}{B_{ik}}=[m]$. We can think of each partition $B_i$ as a smoothing operator, which acts on the entries of row $i$ and changes their value to the expected value of the partition subset they belong to. Formally, the {\em smooth value} of an entry $(i,j)$ such that $j \in B_{ik}$ is defined as
\begin{eqnarray}\label{eq:smooth-value}
A^B_{ij} &=& \frac{\sum_{\ell \in B_{ik}} p_{\ell} \cdot A_{i\ell}}{\sum_{\ell \in B_{ik}} p_{\ell}}.
\end{eqnarray}
Given a partition scheme $B$ that induces the {\em smooth matrix} $A^B$, the resulting {\em partition value} is the expected maximum column entry of $A^B$, namely,
\begin{eqnarray}\label{eq:partition-value}
v^B(A,p)&=&\sum_{j\in [m]}{p_j \cdot \max_i A^B_{ij}}.
\end{eqnarray}
The objective of the asymmetric binary matrix partition problem is to find a partition scheme $B$ such that the resulting partition value $v^B(A,p)$ is maximized.

Alon et al.~\cite{AFGT13} were the first to consider the asymmetric matrix partition problem. They prove that the problem is APX-hard and provide a $0.563$- and a $1/13$-approximation for uniform and non-uniform probability distributions, respectively. They also consider input matrices with non-negative non-binary entries and present a $1/2$- and an $\Omega(1/\log{m})$-approximation algorithm for uniform and non-uniform distributions, respectively. This interesting combinatorial optimization problem has apparent relations to revenue maximization in \textit{take-it-or-leave-it sales}. For example, consider the following setting. There are $m$ items and $n$ potential buyers. Each buyer has a value for each item. Nature selects at random (according to some probability distribution) an item for sale and, then, the seller approaches the highest valuation buyer and offers the item to her at a price equal to her valuation. Can the seller exploit the fact that she has much more accurate information about the items for sale compared to the potential buyers? In particular, information asymmetry arises since the seller knows the realization of the randomly selected item whereas the buyers do not. The approach that is discussed in \cite{AFGT13} is to let the seller define a buyer-specific signalling scheme. That is, for each buyer, the seller can partition the set of items into disjoint subsets (bundles) and report this partition to the buyer. After nature's random choice, the seller can reveal to each buyer the bundle that contains the realization, thus enabling her to update her valuation beliefs. The relation of this problem to asymmetric matrix partition should be apparent. Interestingly, the seller can achieve revenue from items for which no buyer has any value.

This scenario falls within the line of research that studies the impact of information asymmetry to the quality of markets. Akerlof \cite{A70} was the first to introduce a formal analysis of ``markets for lemons'', where the seller has more information than the buyers regarding the quality of the products.
Crawford and Sobel \cite{CS82} study how, in such markets, the seller can exploit her advantage in order to maximize revenue. In \cite{MW82a}, Milgrom and Weber provide the ``linkage principle'' which states that the expected revenue is enhanced when bidders are provided with more information. This principle seems to suggest full transparency but, in \cite{LM10} and \cite{M10} the authors suggest that careful bundling of the items is the best way to exploit information asymmetry.
Many different frameworks that reveal information to the bidders have been proposed in the literature.

More recently, Ghosh et al.~\cite{GNS07} consider full information and propose a clustering scheme according to which, the items are partitioned into bundles and, then, for each such bundle, a separate second-price auction is performed. In this way, the potential buyers cannot bid only for the items that they actually want; they also have to compete for items that they do not have any value for. Hence, the demand for each item is increased and the revenue of the seller is higher. Emek et al.~\cite{EFGPT12} present complexity results in similar settings and Miltersen and Sheffet \cite{MS12} consider fractional bundling schemes for signaling.

In this work, we focus on the simplest binary case of asymmetric matrix partition which has been proved to be APX-hard. We present a $9/10$-approximation algorithm for the uniform case and a $(1-1/e)$-approximation algorithm for non-uniform distributions. Both results significantly improve previous bounds of Alon et al.~\cite{AFGT13}. The analysis of our first algorithm is quite interesting because, despite its purely combinatorial nature, it exploits linear programming techniques. Similar techniques have been used in a series of papers on variants of set cover (e.g. \cite{ACK09,ACKK09,C09,CKK13}) by the second author; however, the application of the technique in the current context requires a quite involved reasoning about the structure of the solutions computed by the algorithm. 

In our second result, we exploit a nice relation of the problem to submodular welfare maximization and use well-known algorithms from the literature. First, we discuss the application of a simple greedy $1/2$-approximation algorithm that has been studied by Lehmann et al.~\cite{LLN06} and then apply Vondr\'ak's smooth greedy algorithm \cite{V08} to achieve a $(1-1/e)$-approximation. Vondr\'ak's algorithm is optimal in the value query model as Khot et al.~\cite{KLMM08} have proved. In a more powerful model where it is assumed that demand queries can be answered efficiently, Feige and Vondr\'ak \cite{FV10} have proved that $(1-1/e+\epsilon)$-approximation algorithms --- where $\epsilon$ is a small positive constant --- are possible. We briefly discuss the possibility/difficulty of applying such algorithms to asymmetric binary matrix partition and observe that the corresponding demand query problems are, in general, NP-hard. 

The rest of the paper is structured as follows. We begin with preliminary definitions and examples in Section \ref{sec:prelim}. Then, we present our $9/10$-approximation algorithm for the uniform case in Section \ref{sec:uniform} and our $(1-1/e)$-approximation algorithm for the non-uniform case in Section \ref{sec:non-uniform}.

\section{Preliminaries} \label{sec:prelim}
Let $A^+=\{j \in [m]:\text{ there exists a row } i \text{ such that } A_{ij}=1\}$ denote the set of columns of $A$ that contain at least one 1-value entry, and $A^0=[m]\backslash A^+$ denote the set of columns of $A$ that contain only 0-value entries.
In the next sections, we usually refer to the sets $A^+$ and $A^0$ as the sets of {\em one-columns} and {\em zero-columns}, respectively.
Furthermore, let  $A_i^+=\{j \in [m]: A_{ij}=1\}$ and $A_i^0=\{j \in [m]: A_{ij}=0\}$ denote the sets of columns that intersect with row $i$ at a 1- and 0-value entry, respectively. All columns in $A_i^+$ are one-columns and, furthermore, $A^+=\cup_{i=1}^n{A_i^+}$. The columns of $A_i^0$ can be either one- or zero-columns and, thus, $A^0 \subseteq \cup_{i=1}^n{A_i^0}$. Also, denote by $r=\sum_{j \in A^+}{p_j}$ the total probability of the one-columns.
As an example, consider the $3\times 6$ matrix
$$A = \left(
\begin{array}{c c c c c c} 
 0 & 1 & 1 & 0 & 1 & 0 \\
 0 & 1 & 1 & 0 & 1 & 0 \\
 0 & 1 & 1 & 0 & 0 & 0 \\
\end{array}\right)$$
and a uniform probability distribution over its columns. We have $A^+=\{2,3,5\}$ and $A^0=\{1,4,6\}$. In the first two rows, the sets $A^+_i$ and $A^0_i$ are identical to $A^+$ and $A^0$, respectively. In the third row, the sets $A^+_3$ and $A^0_3$ are $\{2,3\}$ and $\{1,4,5,6\}$. Finally, the total probability of the one-columns $r$ is $1/2$.

A partition scheme $B$ can be thought of as consisting of $n$ partitions $B_1$, $B_2$, ..., $B_n$ of the set of columns $[m]$. We use the term {\em bundle} to refer to the elements of a partition $B_i$; a bundle is just a non-empty set of columns. For a bundle $b$ of partition $B_i$ corresponding to row $i$, we say that $b$ is an {\em all-zero} bundle if $b\subseteq A^0_i$ and an {\em all-one} bundle if $b\subseteq A^+_i$. A singleton all-one bundle of partition $B_i$ is called {\em column-covering bundle in row} $i$. A bundle that is neither all-zero nor all-one is called {\em mixed}. A mixed bundle corresponds to a set of columns that intersects with row $i$ at both 1- and 0-value entries.

Let us examine the following partition scheme $B$ for matrix $A$ that defines the smooth matrix $A^B$ according to equation (\ref{eq:smooth-value}).
\begin{table}[h!]
\centering
\begin{tabular}{|c|c|}
 \hline 
 $B_1$ & $\{1,2,3,4\}$, $\{5,6\}$ \\ 
 \hline 
 $B_2$ & $\{1,2\}$, $\{3\}$, $\{4,6\}$, $\{5\}$ \\ 
 \hline 
 $B_3$ & $\{1,4,6\}$, $\{2,3,5\}$ \\ 
 \hline 
\end{tabular}\ \ \ \ 
\begin{tabular}{|c|c|c|c|c|c|c|}
\hline
\multirow{3}{*}{$A^B$} &  $1/2$ & $1/2$ & $1/2$ & $1/2$ & $1/2$ & $1/2$ \\
 & $1/2$ & $1/2$ &  $1$  &  $0$  &  $1$  & $0$ \\
 & $0$  &  $2/3$  &  $2/3$  &  $0$  &  $2/3$  &  $0$ \\ 
 \hline
 $\max_i{A_{ij}^B}$ & $1/2$ & $2/3$ & $1$ & $1/2$ & $1$ & $1/2$ \\ 
 \hline
\end{tabular}
\end{table}

\noindent Here, the bundle $\{1,2,3,4\}$ of (the partition $B_1$ of) the first row is a mixed one. The bundle $\{3\}$ of $B_2$ is all-one and, in particular, column-covering in row $2$. The bundle $\{1,4,6\}$ of $B_3$ is all-zero.

By equation (\ref{eq:partition-value}), the partition value is $25/36$ and it can be further improved. First, observe that the leftmost zero-column is included in two mixed bundles (in the first two rows). Also, the mixed bundle in the third row contains a one-column that has been covered through a column-covering bundle in the second row and intersects with the third row at a 0-value entry. Let us modify these two bundles.
\begin{table}[h!]
\centering
\begin{tabular}{|c|c|}
 \hline 
 $B'_1$ & $\{1\}$, $\{2,3,4\}$, $\{5,6\}$ \\ 
 \hline 
 $B'_2$ & $\{1,2\}$, $\{3\}$, $\{4,6\}$, $\{5\}$ \\ 
 \hline 
 $B'_3$ & $\{1,4,5,6\}$, $\{2,3\}$ \\ 
 \hline 
\end{tabular}\ \ \ \ 
\begin{tabular}{|c|c|c|c|c|c|c|}
\hline
\multirow{3}{*}{$A^{B'}$} &  $0$ & $2/3$ & $2/3$ & $2/3$ & $1/2$ & $1/2$ \\
 & $1/2$ & $1/2$ &  $1$  &  $0$  &  $1$  & $0$ \\
 & $0$  &  $1$  &  $1$  &  $0$  &  $0$  &  $0$ \\ 
 \hline
 $\max_i{A_{ij}^{B'}}$ & $1/2$ & $1$ & $1$ & $2/3$ & $1$ & $1/2$ \\ 
 \hline
\end{tabular}
\end{table}

\noindent The partition value becomes $7/9>25/36$. Now, by merging the two mixed bundles $\{2,3,4\}$ and $\{5,6\}$ in the first row, we obtain the smooth matrix below with partition value $47/60>7/9$. Observe that the contribution of column $4$ to the partition value decreases but, overall, we have an increase in the partition value due to the increase in the contribution of column $6$. Actually, such merges never decrease the partition value (see Lemma \ref{lem:structure} below).
\begin{table}[h!]
\centering
\begin{tabular}{|c|c|}
 \hline 
 $B''_1$ & $\{1\}$, $\{2,3,4,5,6\}$ \\ 
 \hline 
 $B''_2$ & $\{1,2\}$, $\{3\}$, $\{4,6\}$, $\{5\}$ \\ 
 \hline 
 $B''_3$ & $\{1,4,5,6\}$, $\{2,3\}$ \\ 
 \hline 
\end{tabular}\ \ \ \ 
\begin{tabular}{|c|c|c|c|c|c|c|}
\hline
\multirow{3}{*}{$A^{B''}$} &  $0$ & $3/5$ & $3/5$ & $3/5$ & $3/5$ & $3/5$ \\
 & $1/2$ & $1/2$ &  $1$  &  $0$  &  $1$  & $0$ \\
 & $0$  &  $1$  &  $1$  &  $0$  &  $0$  &  $0$ \\ 
 \hline
 $\max_i{A_{ij}^{B''}}$ & $1/2$ & $1$ & $1$ & $3/5$ & $1$ & $3/5$ \\ 
 \hline
\end{tabular}
\end{table}

\noindent Finally, by merging the bundles $\{1,2\}$ and $\{3\}$ in the second row and decomposing the bundle $\{2,3\}$ in the last row into two singletons, the partition value becomes $73/90>47/60$ which can be verified to be optimal.  
\begin{table}[h!]
\centering
\begin{tabular}{|c|c|}
 \hline 
 $B'''_1$ & $\{1\}$, $\{2,3,4,5,6\}$ \\ 
 \hline 
 $B'''_2$ & $\{1,2,3\}$, $\{4,6\}$, $\{5\}$ \\ 
 \hline 
 $B'''_3$ & $\{1,4,5,6\}$, $\{2\}$, $\{3\}$ \\ 
 \hline 
\end{tabular}\ \ \ \ 
\begin{tabular}{|c|c|c|c|c|c|c|}
\hline
\multirow{3}{*}{$A^{B'''}$} &  $0$ & $3/5$ & $3/5$ & $3/5$ & $3/5$ & $3/5$ \\
 & $2/3$ & $2/3$ &  $2/3$  &  $0$  &  $1$  & $0$ \\
 & $0$  &  $1$  &  $1$  &  $0$  &  $0$  &  $0$ \\ 
 \hline
 $\max_i{A_{ij}^{B'''}}$ & $2/3$ & $1$ & $1$ & $3/5$ & $1$ & $3/5$ \\ 
 \hline
\end{tabular}
\end{table}

We will now give some more definitions that will be useful in the following. We say that a one-column $j$ is {\em covered} by a partition scheme $B$ if there is at least one row $i$ in which $\{j\}$ is column-covering. For example, in $B'''$, the singleton $\{5\}$ is a column-covering bundle in the second row and the singletons $\{2\}$ and $\{3\}$ are column-covering in the third row. We say that a partition scheme {\em fully covers} the set $A^+$ of one-columns if all of them are covered. In this case, we use the term {\em full cover} to refer to the pairs of indices $(i,j)$ of the 1-value entries $A_{ij}$ such that $\{j\}$ is a column-covering bundle in row $i$. For example, the partition scheme $B'''$ has the full cover $(2,5), (3,2), (3,3)$. 

It turns out that optimal partition schemes always have a special structure like the one of $B'''$.  Alon et al.~\cite{AFGT13} have formalized observations like the above into the following statement.

\begin{lemma}[Alon et al.~\cite{AFGT13}]\label{lem:structure}
Given a uniform instance of the asymmetric binary matrix partition problem with a matrix $A$, there is an optimal partition scheme $B$ with the following properties:
\begin{enumerate}
\item[P1.] $B$ fully covers the set $A^+$ of one-columns.
\item[P2.] For each row $i$, $B_i$ has at most one bundle containing all columns of $A_i^+$ that are not included in column-covering bundles in row $i$ (if any). This bundle can be either all-one (if it does not contain zero-columns) or the unique mixed bundle of row $i$.
\item[P3.] For each zero-column $j$, there exists at most one row $i$ such that $j$ is contained in the mixed bundle of $B_i$ (and $j$ is contained in the all-zero bundles of the remaining rows).
\item[P4.] For each row $i$, the zero-columns that are not contained in the mixed bundle of $B_i$ form an all-zero bundle.
\end{enumerate}
\end{lemma}

Properties P1 and P3 imply that we can think of the partition value as the sum of the contributions of the column-covering bundles and the contributions of the zero-columns in mixed bundles. Property P2 should be apparent; the columns of $A_i^+$ that do not form column-covering bundles in row $i$ are bundled together with zero-columns (if possible) in order to increase the contribution of the latter to the partition value. Property P4 makes $B$ consistent to the definition of a partition scheme where the disjoint union of all the partition subsets in a row should be $[m]$. Clearly, the contribution of the all-zero bundles to the partition value is $0$.
Also, the non-column-covering all-one bundles do not contribute to the partition value either.

The proof of Lemma \ref{lem:structure} in \cite{AFGT13} extensively uses the fact that the instance is uniform. Unfortunately, as we will see later in Section \ref{sec:non-uniform}, the lemma (in particular, property P1) does not hold for non-uniform instances. Luckily, it turns out that non-uniform instances also exhibit some structure which allows us to consider the problem of computing an optimal partition scheme as a {\em welfare maximization problem}. In welfare maximization, there are $m$ items and  $n$ agents; agent $i$ has a valuation function $v_i: 2^{[m]} \rightarrow \mathbb{R}^+$ that specifies her value for each subset of the items. I.e., for a set $S$ of items, $v_i(S)$ represents the value of agent $i$ for $S$. Given a disjoint partition (or allocation) $S=(S_1, S_2, ..., S_n)$ of the items to the agents, where $S_i$ denotes the set of items allocated to agent $i$, the social welfare is the sum of values of the agents for the sets of items allocated to them, i.e., $\text{SW}(S)=\sum_{i}{v_i(S_i)}$. The term welfare maximization refers to the problem of computing an allocation of maximum social welfare. We will discuss only the variant of the problem where the valuations are monotone and submodular; following the literature, we use the term submodular welfare maximization to refer to it.

\begin{definition}
A valuation function $v$ is monotone if $v(S) \leq v(T)$ for any pair of sets $S,T$ such that $S \subseteq T$. A valuation function $v$ is submodular if $v(S \cup \{x\}) - v(S) \geq v(T \cup \{x\}) - v(T)$ for any pair of sets $S,T$ such that $S \subseteq T$ and for any item $x \notin T$.
\end{definition}

An important issue in (submodular) welfare maximization arises with the representation of valuation functions. A valuation function can be described in detail by listing explicitly the values for each of the possible subsets of items. Unfortunately, this is clearly inefficient due to the necessity for exponential input size. A solution that has been proposed in the literature is to assume access to these functions by queries of a particular form. The simplest such form of queries reads as ``what is the value of agent $i$ for the set of items $S$?'' These are known as {\em value queries}. Another type of queries, known as {\em demand queries}, are phrased as follows: ``Given a non-negative price for each item, compute a set $S$ of items for which the difference of the valuation of agent $i$ minus the sum of prices for the items in $S$ is maximized.'' Approximation algorithms that use a polynomial number of valuation or demand queries and obtain solutions to submodular welfare maximization with a constant approximation ratio are well-known in the literature. Our improved approximation algorithm for the non-uniform case of asymmetric binary matrix partition exploits such algorithms.

\section{The uniform case} \label{sec:uniform}
In this section, we present the analysis of a greedy approximation algorithm when the probability distribution $p$ over the columns of the given matrix is uniform. Our algorithm uses a {\em greedy completion procedure} that was also considered by Alon et al.~\cite{AFGT13}. This procedure starts from a full cover of the matrix, i.e., from column-covering bundles in some rows so that all one-columns are covered (by exactly one column-covering bundle). Once this initial full cover is given, the set of columns from $A_i^+$ that are not included in column-covering bundles in row $i$ can form a mixed bundle together with some zero-columns in order to increase the contribution of the latter to the partition value. Greedy completion proceeds as follows. It goes over the zero-columns, one by one, and adds a zero-column to the mixed bundle of the row that maximizes the marginal contribution of the zero-column. The marginal contribution of a zero-column to the partition value when it is added to a mixed bundle that consists of $x$ zero-columns and $y$ one-columns is given by the quantity
\begin{eqnarray*}
\Delta(x,y)&=&(x+1)\frac{y}{x+y+1}-x\frac{y}{x+y} = \frac{y^2}{(x+y)(x+y+1)}.
\end{eqnarray*}
The right-hand side of the first equality is simply the difference between the contribution of $x+1$ and $x$ zero-columns to the partition value when they form a mixed bundle with $y$ one-columns. Alon et al.~\cite{AFGT13} proved that, in the uniform case, this greedy completion procedure yields the maximum contribution from the zero-columns to the partition value among all partition schemes that include a given full cover. We extensively use this property as well as the fact that $\Delta(x,y)$ is non-decreasing with respect to $y$.

So, our algorithm consists of two phases. In the first phase, called the {\em cover phase}, the algorithm computes an arbitrary full cover for set $A^+$. In the second phase, called the {\em greedy phase}, it simply runs the greedy completion procedure mentioned above. In the rest of this section, we will show that this simple algorithm obtains an approximation ratio of $9/10$; we will also show that our analysis is tight. Even though our algorithm is purely combinatorial, our analysis exploits linear programming duality.

Overall, the partition value obtained by the algorithm can be thought of as the sum of contributions from column-covering bundles (this is exactly $r$) plus the contribution from the mixed bundles created during the greedy phase (i.e., the contribution from the zero-columns). Denote by $\rho$ the ratio between the total number of appearances of one-columns in the mixed bundles of the optimal partition scheme (so, the number each one-column is counted equals the number of mixed bundles that contain it) and the number of zero-columns. For example, in the partition scheme $B'''$ in the example of the previous section, the two mixed bundles are $\{2,3,4,5,6\}$ in the first row and $\{1,2,3\}$ in the second row. So, the one-columns $2$ and $3$ appear twice while the one-column $5$ appears once in these mixed bundles. Since we have three zero-columns, the value of $\rho$ is $5/3$. We can use the quantity $\rho$ to upper-bound the optimal partition value as follows.

\begin{lemma}\label{lem:opt}
The optimal partition value is at most $r+(1-r)\frac{\rho}{\rho+1}$.
\end{lemma}

\begin{proof}
The first term in the above expression represents the contribution of the one-columns in the full cover of the optimal partition scheme. The second term is the hypothetical contribution of all zero-columns assuming that all mixed bundles of the optimal partition scheme are merged into one. The proof follows by an observation of Alon et al.~\cite[Lemma 1]{AFGT13} that has been used to prove what we state as property P2 in Lemma \ref{lem:structure}. According to that observation, the contribution of the zero-columns of two disjoint bundles to the partition value is not higher than their contribution when the two disjoint bundles are merged.
\end{proof}

In our analysis, we distinguish between two main cases depending on the value of $\rho$. The first case is when $\rho < 1$; in this case, the additional partition value obtained during the greedy phase of the algorithm is lower-bounded by the partition value we would have by creating bundles containing exactly one one-column and either $\lceil 1/\rho \rceil$ or $\lfloor 1/\rho \rfloor$ zero-columns.

\begin{lemma}\label{lem:small-rho}
If $\rho<1$, then the partition value obtained by the algorithm is at least $0.97$ times the optimal one.
\end{lemma}

\begin{proof}
We will lower-bound the partition value returned by the algorithm by considering the following formation of mixed bundles as an alternative to the greedy completion procedure used in the greedy phase. For each appearance of a one-column in a mixed bundle in the partition scheme computed by the algorithm (observe that the total number of such appearances is exactly $\rho m(1-r)$), we include this one-column in a mixed bundle together with either $\ceil{1/\rho}$ or $\floor{1/\rho}$ distinct zero-columns. By Lemma \ref{lem:opt}, this process yields an optimal partition value if $1/\rho$ is an integer. Otherwise, denote by $x=m(1-r)(1-\rho\floor{1/\rho})$ the number of mixed bundles containing $\ceil{1/\rho}$ zero-columns. Then, the number of mixed bundles containing $\floor{1/\rho}$ zero-columns will be $\rho m(1-r) -x = m(1-r)(\rho\ceil{1/\rho}-1)$. Observe that the smooth value of a zero-column is $\frac{1}{1+\ceil{1/\rho}}$ in the first case and $\frac{1}{1+\floor{1/\rho}}$ in the second case. Hence, we can bound the partition value obtained by the algorithm as follows:
\begin{align*}
\mbox{ALG} \geq r + (1-r)(1-\rho\floor{1/\rho}) \frac{\ceil{1/\rho}}{1+\ceil{1/\rho}} + (1-r)(\rho\ceil{1/\rho}-1)\frac{\floor{1/\rho}}{1+\floor{1/\rho}}.
\end{align*}
Now, assuming that $\rho \in (\frac{1}{k+1},\frac{1}{k})$ for some integer $k \geq 1$, we have that $\floor{1/\rho}=k$ and $\ceil{1/\rho}=k+1$ and, hence, 
\begin{align*}
\mbox{ALG} \geq r +(1-r)\frac{1+\rho k(k+1)}{(k+1)(k+2)}.
\end{align*}
Using Lemma \ref{lem:opt}, we have
\begin{align*}
\frac{\mbox{ALG}}{\mbox{OPT}} &\geq \frac{r +(1-r)\frac{1+\rho k(k+1)}{(k+1)(k+2)}}{r+ (1-r)\frac{\rho}{\rho+1}} \geq \frac{\frac{1+\rho k(k+1)}{(k+1)(k+2)}}{\frac{\rho}{\rho+1}} = \frac{(1+1/\rho)(1+\rho k(k+1))}{(k+1)(k+2)}.
\end{align*}
This last expression is minimized (with respect to $\rho$) for $1/\rho=\sqrt{k(k+1)}$. Hence,
\begin{align*}
\frac{\mbox{ALG}}{\mbox{OPT}} &\geq \frac{\left(1+\sqrt{k(k+1)}\right)^2}{(k+1)(k+2)},
\end{align*}
which is minimized for $k=1$ to approximately $0.97$.
\end{proof}

For the case $\rho \geq 1$, we use completely different arguments. We will reason about the solution produced by the algorithm by considering a particular decomposition of the set of mixed bundles computed in the greedy phase. Then, using again the observation of Alon et al.~\cite[Lemma 1]{AFGT13}, the contribution of the zero-columns to the partition value in the solution computed by the algorithm is lower-bounded by their contribution to the partition value when that are part of the mixed bundles obtained after the decomposition.

The decomposition is defined as follows. It takes as input a bundle with $y$ zero-columns and $x$ one-columns and decomposes it into $y$ bundles containing exactly one zero-column and either $\lfloor x/y \rfloor$ or $\lceil x/y \rceil$ one-columns. Note that if $x/y$ is not an integer, there will be $x-y\lfloor x/y \rfloor$ bundles with $\lceil x/y \rceil$ one-columns. The solution obtained after the decomposition of the solution returned by the algorithm has a very special structure as our next lemma suggests.

\begin{lemma}\label{lem:s-t}
There exists an integer $s \geq 1$ such that each bundle in the decomposition has at least $s$ and at most $3s$ one-columns.
\end{lemma}

\begin{proof}
Consider the application of the decomposition step to the mixed bundles that are computed by the algorithm and let $s$ be the minimum number of one-columns among the decomposed mixed bundles.
This implies that one of the mixed bundles, say $b_1$, computed by the algorithm has $\mu$ zero-columns and at most $(s+1)\mu-1$ one-columns.
Denoting by $\nu$ the number of one-columns in this bundle, we have that the marginal partition value when the last zero-column $Z$ is included in $b_1$ is exactly
\begin{align*}
\Delta(\mu,\nu) = \frac{\nu^2}{(\nu+\mu)(\nu+\mu-1)} \leq \frac{( (s+1)\mu-1 )^2 }{( (s+2)\mu-1 ) ( (s+2)\mu-2 )}
\end{align*} 
since $\Delta(\mu,\nu)$ is increasing in $\nu$ and $\nu \leq (s+1)\mu-1$. The rightmost expression is decreasing in $\mu$ and $\mu\geq 1$; hence, the marginal partition value of $Z$ is at most $\frac{s}{s+1}$.

Now assume for the sake of contradiction that one of the mixed bundles obtained after the decomposition has at least $3s+1$ one-columns. Clearly, this must have been obtained by the decomposition of a mixed bundle $b_2$ (returned by the algorithm) with $\lambda$ zero-columns and at least $(3s+1)\lambda$ one-columns.
Denote by $\nu'$ the number of one-columns in this bundle and let us compute the marginal partition value if the zero-column $Z$ would be included in $b_2$. This would be 
\begin{align*}
\Delta(\lambda+1,\nu') &= \frac{\nu'^2}{(\nu'+\lambda+1)(\nu'+\lambda)} \geq \frac{(3s+1)^2\lambda}{((3s+2)\lambda+1)(3s+2)} \geq \frac{(3s+1)^2}{(3s+3)(3s+2)}.
\end{align*} 
The first inequality follows since the marginal partition value function is increasing in $\nu'$ and $\nu'\geq (3s+1)\lambda$, and the second one follows since $\lambda \geq 1$. 
Now, the last quantity can be easily verified to be strictly higher that $\frac{s}{s+1}$ and the algorithm should have included $Z$ in $b_2$ instead. We have reached the desired contradiction that proves the lemma.
\end{proof}

Now, our analysis proceeds as follows. For every triplet $r \in [0,1], \rho \geq 1$ and integer $s\geq 1$, we will prove that any solution consisting of an arbitrary cover of the $rm$ one-columns and the decomposed set of bundles containing at least $s$ and at most $3s$ one-columns yields a $9/10$-approximation of the optimal partition value. By the discussion above, this will also be the case for the solution returned by the algorithm.
In order to account for the worst-case contribution of zero-columns to the partition value for a given triplet of parameters, we will use the following linear program, which we denote by $\text{LP}(r,\rho,s)$:
\begin{eqnarray*} \label{eq:lp}
&\mbox{minimize } & \sum_{k=s}^{3s}{ \frac{k}{k+1} \theta_k } \\ 
&\mbox{subject to: } & \sum_{k=s}^{3s}{\theta_k}=1-r \\ 
& & \sum_{k=s}^{3s}{k \theta_k} \geq \rho(1-r)-r \\ 
& & \theta_k \geq 0, k=s, ..., 3s
\end{eqnarray*}

The variable $\theta_k$ denotes the total probability of the zero-columns that participate in decomposed mixed bundles with $k$ one-columns. The objective is to minimize the contribution of the zero-columns to the partition value. The equality constraint means that all zero-columns have to participate in bundles. The inequality constraint requires that the total number of appearances of one-columns in bundles used by the algorithm is at least the total number of appearances of one-columns in mixed bundles of the optimal partition scheme minus one appearance for each one-column, since for every selection of the cover, the algorithm will have the same number of (appearances of) one-columns available to form mixed bundles. Informally, the linear program answers (rather pessimistically) to the question of how inefficient the algorithm can be. In particular, given an instance with parameters $r$ and $\rho$, the quantity $\min_{\text{\normalfont int } s\geq 1}\mbox{\normalfont LP}(r,\rho,s)$ yields a lower bound on the contribution of the zero-columns to the partition value and $r+\min_{\text{\normalfont int }  s\geq 1}\mbox{\normalfont LP}(r,\rho,s)$ is a lower bound on the partition value. The next lemma completes the analysis of the greedy algorithm for the case $\rho \geq 1$.

\begin{lemma}\label{lem:big-rho}
For every $r \in [0,1]$ and $\rho \geq 1$,
$$r+ \min_{\text{\normalfont int } s\geq 1} \mbox{\normalfont LP}(r,\rho,s) \geq \frac{9}{10} \text{\normalfont OPT}.$$
\end{lemma}

\begin{proof}
We will prove the lemma using LP-duality.
The dual of $\text{LP}(r,\rho,s)$ is:
\begin{eqnarray*} \label{eq:dlp}
&\mbox{maximize } & (1-r)\alpha + ((1-r)\rho-r))\beta \\ 
&\mbox{subject to: } & k\beta + \alpha \leq \frac{k}{k+1}, k=s, ..., 3s \\ 
& & \beta \geq 0
\end{eqnarray*}
Using Lemma \ref{lem:opt}, we bound the optimal partition value as
\begin{align*}
\mbox{OPT} \leq r + (1-r)\frac{\rho}{\rho+1} = \frac{\rho+r}{\rho+1}.
\end{align*}
Hence, it suffices to show that, for every triplet of parameters $(r, \rho, s)$, there is a feasible dual solution of objective value $D(r,\rho,s)$ that satisfies
\begin{align} \label{eq:rD}
r + D(r,\rho,s) - \frac{9}{10}\frac{\rho+r}{\rho+1} & \geq 0.
\end{align}
The feasible region of the dual is defined by the lines $\beta=0$, $\alpha=\frac{s}{s+1}-s\beta$ and $\alpha=\frac{3s}{3s+1}-3s\beta$; the remaining constraints can be easily seen to be redundant. The two important intersections of those lines are the points 
$$(\alpha,\beta)=\left(\frac{s}{s+1},0\right) \mbox{ and } (\alpha,\beta)=\left(\frac{3s^2}{(s+1)(3s+1)},\frac{1}{(s+1)(3s+1)}\right)$$ with objective values 
$$D_1(r,\rho,s)=\frac{s}{s+1}(1-r)\, \mbox{ and } D_2(r,\rho,s)=\frac{3s^2}{(s+1)(3s+1)}(1-r)+\frac{\rho(1-r)-r}{(s+1)(3s+1)},$$ respectively.
We will show that one of these two points can always be used as a feasible dual solution in order to prove inequality (\ref{eq:rD}). We distinguish between two cases.

\paragraph{Case I:} $r \geq \frac{\rho-1}{\rho}$. We will show that the point with dual objective value $D_1(r,\rho,s)$ satisfies inequality (\ref{eq:rD}), i.e., 
\begin{align}\label{eq:ineq1}
r + \frac{s}{s+1}(1-r) - \frac{9}{10}\frac{\rho+r}{\rho+1} & \geq 0.
\end{align}
Since $s\geq 1$, we have that the left hand side of inequality (\ref{eq:ineq1}) is at least
\begin{align*}
\frac{1+r}{2}-\frac{9}{10}\frac{\rho+r}{\rho+1} = \frac{1}{2}-\frac{9\rho}{10(\rho+1)} +r \left( \frac{1}{2} - \frac{9}{10(\rho+1)} \right).
\end{align*}
Since $\rho\geq 1$, we have that $\frac{1}{2} - \frac{9}{10(\rho+1)} \geq 0$, and we can lower-bound the above quantity using the assumption $r \geq \frac{\rho-1}{\rho}$, as follows:
\begin{align*}
\frac{1+r}{2} - \frac{9}{10}\frac{\rho+r}{\rho+1} &\geq \frac{1}{2}-\frac{9\rho}{10(\rho+1)} +\frac{\rho - 1}{\rho} \left( \frac{1}{2} - \frac{9}{10(\rho+1)} \right)
= \frac{(\rho-2)^2}{10\rho(\rho+1)}
\geq 0,
\end{align*}
and inequality (\ref{eq:ineq1}) follows.

\paragraph{Case II:} $r < \frac{\rho-1}{\rho}$. We will now show that the point with dual objective value $D_2(r,\rho,s)$ satisfies inequality (\ref{eq:rD}), i.e., 
\begin{align}\label{eq:ineq2}
r + \frac{3s^2}{(s+1)(3s+1)}(1-r)+\frac{\rho(1-r)-r}{(s+1)(3s+1)} -\frac{9}{10}\frac{\rho+r}{\rho+1} &\geq 0.
\end{align}
Let us denote by $F$ the left hand side of inequality (\ref{eq:ineq2}). With simple calculations, we obtain
\begin{align}\label{eq:F}
F &= \frac{10\rho^2-(-3s^2+36s-1)\rho+30s^2}{10(3s+1)(s+1)(\rho+1)}
- r \cdot \frac{10\rho^2-(40s-10)\rho+27s^2-4s+9}{10(3s+1)(s+1)(\rho+1)}.
\end{align}
Observe that the numerator of the left fraction in (\ref{eq:F}) is a quadratic function with respect to $\rho$ with positive coefficient in the leading term. Its discriminant is $-1191s^4-216s^3+1296s^2-72 s +7$ which is clearly negative for every integer $s\geq 1$. Hence, the numerator of the left fraction is always positive. Now, if the numerator of the rightmost fraction is negative, then inequality (\ref{eq:ineq2}) is obviously satisfied. Otherwise, using the assumption $r <\frac{\rho-1}{\rho}$, we have
\begin{align*}
F &\geq \frac{10\rho^2-(-3s^2+36s-1)\rho+30s^2}{10(3s+1)(s+1)(\rho+1)}
- \frac{\rho-1}{\rho} \cdot \frac{10\rho^2-(40s-10)\rho+27s^2-4s+9}{10(3s+1)(s+1)(\rho+1)}\\
&= \frac{(3s^2+4s+1)\rho^2+(3s^2-36s+1)\rho+27s^2-4s+9}{10\rho(3s+1)(s+1)(\rho+1)}.
\end{align*}
Now, the numerator of the last fraction is again a quadratic function in terms of $\rho$ with positive coefficient in the leading term and discriminant equal to
$$-315 s^4-600 s^3+1150 s^2-200 s-35 = (-315 s^3-915 s^2+235 s -35)(s-1) \leq 0,$$
for every integer $s\geq 1$. Hence, $F\geq 0$ and the proof is complete.
\end{proof}

The next statement summarizes the discussion above.

\begin{theorem}\label{thm:uniform-main}
The greedy algorithm always yields a $9/10$-approximation of the optimal partition value in the uniform case.
\end{theorem}

Our analysis is tight as our next counter-example suggests.

\begin{theorem}
There exists an instance of the uniform asymmetric binary matrix partition problem for which the greedy algorithm computes a partition scheme with value (at most) $9/10$ of the optimal one.
\end{theorem}

\begin{proof}
Consider the instance of the asymmetric binary matrix partition problem that consists of the matrix
$$A = \left(\begin{array}{c c c c}1&0&0&0\\0&1&0&0\\1&1&0&0\\1&1&0&0\end{array}\right)$$
with $p_i=1/4$ for $i=1, 2, 3, 4$. The optimal partition value is obtained by covering the one-columns in the first two rows and then bundling each of the two zero-columns with a pair of one-columns in the third and fourth row, respectively. This yields a partition value of $5/6$. The greedy algorithm may select to cover the one-columns using the 1-value entries $A_{31}$ and $A_{42}$. This is possible since the greedy algorithm has no particular criterion for breaking ties when selecting the full cover. Given this full cover, the greedy completion procedure will assign each of the two zero-columns with only one one-column. The partition value is then $3/4$, i.e., $9/10$ times the optimal partition value.
\end{proof}

\section{Asymmetric binary matrix partition as welfare maximization} \label{sec:non-uniform}
We now consider the more general non-uniform case. Interestingly, property P1 of Lemma \ref{lem:structure} does not hold any more as the following statement shows.

\begin{lemma}
For every $\epsilon>0$, there exists an instance of the asymmetric binary matrix partition problem in which any partition scheme containing a full cover of the columns in $A^+$ yields a partition value that is at most $8/9+\epsilon$ times the optimal one.
\end{lemma}

\begin{proof}
Consider the instance of the asymmetric binary matrix partition problem consisting of the matrix 
$$A = \left(\begin{array}{c  c  c  c}1&0&0&0\\0&1&0&0\\0&1&0&0\\1&0&1&0\end{array}\right)$$
with column probabilities $p_j=\frac{1}{\beta+3}$ for $j=1,2,3$ and $p_4=\frac{\beta}{\beta+3}$ for $\beta > 2$. Observe that there are four partition schemes containing a full cover (depending on the rows that contain the column-covering bundle of the first two columns) and, in each of them, the zero-column is bundled together with a 1-value entry. By making calculations, we obtain that the partition value in these cases is $\frac{4\beta+3}{(\beta+1)(\beta+3)}$. Here is one of these partition schemes:

\begin{table}[h!]
\centering
\begin{tabular}{|c|c|}
 \hline 
 $B_1$ & $\{1\}$, $\{2,3,4\}$ \\ 
 \hline 
 $B_2$ & $\{2\}$, $\{1,3,4\}$ \\ 
 \hline 
 $B_3$ & $\{1,3\}$, $\{2,4\}$ \\ 
 \hline 
 $B_4$ & $\{1\}$, $\{3\}$, $\{2,4\}$ \\ 
 \hline  
\end{tabular}\ \ \ \ 
\begin{tabular}{|c|c|c|c|c|}
\hline
\multirow{4}{*}{$A^B$} &  $1$ & $0$ & $0$ & $0$ \\
 & $0$ & $1$ &  $0$  &  $0$ \\
 & $0$  &  $\frac{1}{\beta+1}$  &  $0$  &  $\frac{1}{\beta+1}$ \\ 
 & $1$ & $0$ & $1$ & $0$ \\
 \hline
 $p_j \cdot \max_i{A_{ij}^B}$ & $\frac{1}{\beta+3}$ & $\frac{1}{\beta+3}$ & $\frac{1}{\beta+3}$ & $\frac{\beta}{(\beta+1)(\beta+3)}$ \\ 
 \hline
\end{tabular}
\end{table}

In contrast, consider the partition scheme $B'$ in which the 1-value entries $A_{11}$ and $A_{22}$ form column-covering bundles in rows $1$ and $2$, the entries $A_{32}$ and $A_{33}$ are bundled together in row $3$ and the entries $A_{41}$, $A_{43}$, and $A_{44}$ are bundled together in row $4$. As it can be seen from the tables below (recall that $\beta>2$), the partition value now becomes $\frac{4.5\beta+5}{(\beta+2)(\beta+3)}$. 

\begin{table}[h!]
\centering
\begin{tabular}{|c|c|}
 \hline 
 $B'_1$ & $\{1\}$, $\{2,3,4\}$ \\ 
 \hline 
 $B'_2$ & $\{2\}$, $\{1,3,4\}$ \\ 
 \hline 
 $B'_3$ & $\{1,4\}$, $\{2,3\}$ \\ 
 \hline 
 $B'_4$ & $\{2\}$, $\{1,3,4\}$ \\ 
 \hline  
\end{tabular}\ \ \ \ 
\begin{tabular}{|c|c|c|c|c|}
\hline
\multirow{4}{*}{$A^{B'}$} &  $1$ & $0$ & $0$ & $0$ \\
 & $0$ & $1$ &  $0$  &  $0$ \\
 & $0$  &  $1/2$  &  $1/2$  &  $0$ \\ 
 & $\frac{2}{\beta+2}$ & $0$ & $\frac{2}{\beta+2}$ & $\frac{2}{\beta+2}$ \\
 \hline
 $p_j \cdot \max_i{A_{ij}^{B'}}$ & $\frac{1}{\beta+3}$ & $\frac{1}{\beta+3}$ & $\frac{1}{2(\beta+3)}$ & $\frac{2\beta}{(\beta+2)(\beta+3)}$  \\ 
 \hline
\end{tabular}
\end{table}

Clearly, the ratio of the two partition values approaches $8/9$ from above as $\beta$ tends to infinity. Hence, the theorem follows by selecting $\beta$ sufficiently large for any given $\epsilon>0$.
\end{proof}

Still, as the next statement indicates, the optimal partition scheme has some structure which we will exploit later.

\begin{lemma}\label{lem:non-uniform-structure}
Consider an instance of the asymmetric binary matrix partition problem consisting of a matrix $A$ and a probability distribution $p$ over its columns. There is an optimal partition scheme $B$ that satisfies properties P2, P3, P4 as well as the following property: 
\begin{enumerate}
\item[P5.] Given any column $j$, denote by $H_j = \arg\max_i{A_{ij}^B}$ the set of rows through which column $j$ contributes to the partition value $v^B(A,p)$. For every $i\in H_j$ such that $A_{ij}=1$, the bundle of partition $B_i$ that contains column $j$ is not mixed.
\end{enumerate}
\end{lemma}

\begin{proof}
Consider an optimal partition scheme $B$ that does not satisfy property P5, and let $j^*$ be a column such that $A_{i^*j^*}=1$ for some $i^*\in H_{j^*}$. Furthermore, assume that the mixed bundle $b$ of partition $B_{i^*}$ that contains column $j^*$, also contains the columns of a (possibly empty) set $b_1 \subseteq A_{i^*}^+\setminus\{j^*\}$ and the columns of a non-empty set $b_0 \subseteq A_{i^*}^0$. Let $p^+\geq 0$ and $p^0>0$ be the sum of probabilities of the columns in $b_1$ and $b_0$, respectively.

Let $B'$ be the partition scheme that is obtained from $B$ when splitting bundle $b$ into two bundles $\{j^*\}$ and $b\setminus \{j^*\}$; we will show that $B'$ must be optimal as well. Observe that $A_{i^*j}^B=\frac{p_{j^*}+p^+}{p_{j^*}+p^++p^0}$ and $A_{i^*j}^{B'}=\frac{p^+}{p^++p^0}$ for every $j\in b\setminus\{j^*\}$; hence, $A_{i^*j}^B>A_{i^*j}^{B'}$. Since, this is the only difference between $B$ and $B'$, the difference $\max_i{A_{ij}^B} - \max_i{A_{ij}^{B'}}$ is at most $A_{i^*j}^B-A_{i^*j}^{B'}$ for every $j\in b\setminus\{j^*\}$, and
$\max_i{A_{ij^*}^B} - \max_i{A_{ij^*}^{B'}} = A_{i^*j^*}^B-A_{i^*j^*}^{B'} = \frac{p_{j^*}+p^+}{p_{j^*}+p^++p^0}-1$. Hence, we have
\begin{align*}
v^B(A,p) - v^{B'}(A,p) &= \sum_{j\in [m]}{p_j \cdot \max_i{A_{ij}^B}} - \sum_{j\in [m]}{p_j \cdot \max_i{A_{ij}^{B'}}}\\
&= \sum_{j\in b}{p_j\left(\max_i{A_{ij}^B} - \max_i{A_{ij}^{B'}}\right)}\\
&\leq  \sum_{j\in b}{p_j\left(A_{i^*j}^B-A_{i^*j}^{B'}\right)}\\
&= p_{j^*}\left(\frac{p_{j^*}+p^+}{p_{j^*}+p^++p^0}-1\right) + \sum_{j\in b\setminus\{j^*\}}{p_j\left(\frac{p_{j^*}+p^+}{p_{j^*}+p^++p^0}-\frac{p^+}{p^++p^0}\right)}\\
&= 0,
\end{align*}
where the last equality follows from the fact that $\sum_{j\in b\setminus\{j^*\}}{p_j}=p^++p^0$. Hence, the partition value does not decrease. By repeating this argument, we will reach an optimal partition scheme that satisfies property P5. Using arguments similar to the ones used in the proof of Alon et al.~\cite{AFGT13} for Lemma~\ref{lem:structure} we can prove that the resulting partition scheme can be transformed in such a way so that it satisfies properties P2, P3, and P4.
\end{proof}

What Lemma \ref{lem:non-uniform-structure} says is that the contribution of column $j\in A^+$ to the partition value comes from a row $i$ such that either $j\in A^+_i$ and $\{j\}$ forms a column-covering bundle or $j\in A^0_i$ and $j$ belongs to the mixed bundle of row $i$. The contribution of a column $j\in A^0$ to the partition value always comes from a row $i$ where $j$ belongs to the mixed bundle. Hence, the problem of computing the partition scheme of optimal partition value is equivalent to deciding the row from which each column contributes to the partition value.

Let $B$ be a partition scheme and $S$ be a set of columns whose contribution to the partition value of $B$ comes from row $i$ (i.e., $i$ is the row that maximizes the smooth value $A_{ij}^B$ for each column $j$ in $S$). Denoting the sum of these contributions by $R_i(S)=\sum_{j\in S}{p_j \cdot A_{ij}^B}$, we can equivalently express $R_i(S)$ as
\begin{align*}
R_i(S)=\sum_{j \in S \cap A_i^+}{p_j} + \frac{\sum_{j \in S \cap A_i^0}{p_j \sum_{j \in A_i^+ \setminus S}{p_j}}}{\sum_{j \in S \cap A_i^0} p_j + \sum_{j \in A_i^+ \setminus S}{p_j}}.
\end{align*}
The first sum represents the contribution of columns of $S\cap A^+_i$ to the partition value (through column-covering bundles) while the second sum represents the contribution of the columns in $S\cap A^0_i$ which are bundled together with all $1$-value entries in $A^+_i\setminus S$ in the mixed bundle of row $i$. Then, the partition scheme $B$ can be thought of as a collection of disjoint sets $S_i$ (with one set per row) such that $S_i$ contains those columns whose entries achieve their maximum smooth value in row $i$. Hence, the partition value of $B$ is $v^B(A,p) = \sum_{i\in[n]}{R_i(S_i)}$ and the problem is essentially equivalent to welfare maximization where the rows act as the agents who will be allocated bundles of items (corresponding to columns).

\begin{lemma}\label{lem:submodular}
For every row $i$, the function $R_i$ is non-decreasing and submodular.
\end{lemma}

\begin{proof}
We will show that the function $R_i$ is non-decreasing and has decreasing marginal utilities, i.e.,
\begin{itemize}
\item (monotonicity) for every set $S$ and item $x\not\in S$, it holds that $R_i(S)\leq R_i(S\cup \{x\})$;
\item (decreasing marginal utilities) for every pair of sets $S, T$ such that $S \subseteq T$ and every item $x \not\in T$, it holds that $R_i(S \cup \{x\})-R_i(S) \geq R_i(T \cup \{x\})-R_i(T)$.
\end{itemize}
In order to simplify notation, let us define the functions $\alpha(S)=\sum_{j \in S \cap A_i^+}p_j$, $\beta(S)=\sum_{j \in S \cap A_i^0} p_j$ and $\gamma(S)=\sum_{j \in A_i^+ \setminus S} p_j$. We can rewrite the function $R_i$ as
\begin{align*}
R_i(S)=\alpha(S) + \frac{\beta(S) \cdot \gamma(S)}{\beta(S) + \gamma(S)}.
\end{align*}

Let $S, T \subseteq [m]$ be two sets of columns such that $S\subseteq T$ and let $x$ be a column that does not belong to set $T$. We distinguish between two cases depending on $x$. If $x\in A_i^+$, observe that
\begin{itemize}
\item $\alpha(S \cup \{x\}) = \alpha(S)+p_x$ and $\alpha(T \cup \{x\}) = \alpha(T)+p_x$;
\item $\beta(S \cup \{x\}) = \beta(S)$ and $\beta(T \cup \{x\}) = \beta(T)$;
\item $\gamma(S \cup \{x\}) = \gamma(S)-p_x$ and $\gamma(T \cup \{x\}) = \gamma(T)-p_x$.
\end{itemize}
Using the definition of function $R_i$, we have
\begin{align*}
R_i(S \cup \{x\})- R_i(S) &= p_x + \beta(S) \left( \frac{\gamma(S)-p_x}{\beta(S)+\gamma(S)-p_x} - \frac{\gamma(S)}{\beta(S)+\gamma(S)} \right) \\
&= p_x - \frac{ p_x \beta(S)^2}{(\beta(S)+\gamma(S))\cdot(\beta(S)+\gamma(S)-p_x)} \\
&\geq p_x - \frac{ p_x \beta(S)^2}{(\beta(S)+\gamma(T))(\beta(S)+\gamma(T)-p_x)} \\
&\geq p_x - \frac{ p_x \beta(T)^2}{(\beta(T)+\gamma(T))(\beta(T)+\gamma(T)-p_x)} \\
&= R_i(T \cup \{x\})- R_i(T).
\end{align*}
The first inequality follows since $\gamma$ is clearly non-increasing and $S \subseteq T$ and the second inequality follows by applying twice (with $a=\gamma(T)$ and $a=\gamma(T)-p_x$, respectively) the fact that the function $f(z)=\frac{z}{z+a}$ for $a\geq 0$ is non-decreasing.

If instead $x\in A_i^0$, observe that
\begin{itemize}
\item $\alpha(S \cup \{x\}) = \alpha(S)$ and $\alpha(T \cup \{x\}) = \alpha(T)$;
\item $\beta(S \cup \{x\}) = \beta(S)+p_x$ and $\beta(T \cup \{x\}) = \beta(T)+p_x$;
\item $\gamma(S \cup \{x\}) = \gamma(S)$ and $\gamma(T \cup \{x\}) = \gamma(T)$.
\end{itemize}
Hence, we have
\begin{align*}
R_i(S \cup \{x\})- R_i(S) &= \gamma(S) \left( \frac{\beta(S)+p_x}{\beta(S)+\gamma(S)+p_x} - \frac{\beta(S)}{\beta(S)+\gamma(S)} \right) \\
&= \frac{ p_x \gamma(S)^2}{(\beta(S)+\gamma(S))(\beta(S)+\gamma(S)+p_x)} \\
&\geq \frac{ p_x \gamma(S)^2}{(\beta(T)+\gamma(S))(\beta(T)+\gamma(S)+p_x)} \\
&\geq \frac{ p_x \gamma(T)^2}{(\beta(T)+\gamma(T))(\beta(T)+\gamma(T)+p_x)}\\
&= R_i(T \cup \{x\})- R_i(T).
\end{align*}
Again, the inequality follows since $\beta$ is clearly non-decreasing and $S\subseteq T$ and the second inequality follows by applying twice (with $a=\beta(T)$ and $a=\beta(T)+p_x$, respectively) the fact that the function $f(z)=\frac{z}{z+a}$ with $a\geq 0$ is non-decreasing.

We have completed the proof that $R_i$ has decreasing marginal utilities. In order to establish monotonicity, it suffices to observe that the quantity at the right-hand side of the second equality in each of the above two derivations starting with $R_i(S\cup \{x\})-R_i(S)$ is non-negative.
\end{proof}

Lehmann et al.~\cite{LLN06} studied the submodular welfare maximization problem and provided a simple algorithm that yields a $1/2$-approximation of the optimal welfare. Their algorithm considers the items one by one and assigns item $j$ to the agent that maximizes the marginal valuation (the additional value from the allocation of item $j$). In our setting, this algorithm can be implemented as follows. It considers the one-columns first and the zero-columns afterwards. Whenever considering a one-column $j$, a column-covering bundle $\{j\}$ is formed at an arbitrary row $i$ with $j\in A^+_i$ (such a decision definitely maximizes the increase in the partition value). Whenever considering a zero-column, it includes it to a mixed bundle so that the increase in the partition value is maximized. Using the terminology of Alon et al.~\cite{AFGT13}, the algorithm essentially starts with an arbitrary cover of the one-columns and then it runs the greedy completion procedure. Again, we will use the term greedy for this algorithm. 

\begin{theorem}\label{thm:lehmann}
The greedy algorithm for the asymmetric binary matrix partition problem has approximation ratio at least $1/2$. This bound is tight.
\end{theorem}

\begin{proof}
The lower bound holds by the equivalence of the greedy algorithm with the algorithm studied by Lehmann et al.~\cite{LLN06}. Below, we prove the upper bound. In particular, we show that for every $\epsilon>0$, there exists an instance of the problem in which the greedy algorithm obtains a partition scheme whose value is at most $1/2+\epsilon$ of the optimal one.

Let $k>0$ be a positive integer and $\alpha$ significantly higher than $k$. Consider the instance of the asymmetric binary matrix partition that consists of the following $(k+1) \times (k+1)$ matrix
$$A = \left(\begin{array}{c c c c c}1&0&\cdots&0&0\\0&1&\cdots&0&0\\\vdots&\vdots&\ddots&\vdots&\vdots\\0&0&\cdots&1&0\\1&1&\cdots&1&0\end{array}\right)$$ 
where $p_j=\frac{1}{k+\alpha}$ for $j\in [k]$ and $p_{k+1}=\frac{\alpha}{k+\alpha}$. So, the first $k$ columns and rows of $A$ form an identity matrix, the last column has only 0-value entries and the last row consists of $k$ 1-value entries in the first $k$ columns. In order to lower-bound the optimal partition value, consider the partition scheme consisting of a full cover that contains the 1-value entries $(i,i)$ for $i\leq k$, and a bundle containing the whole $(k+1)$-th row. The optimal partition value is lower-bounded by the value of this partition scheme. By simple calculations, we obtain
\begin{align*}
\mbox{OPT} &\geq \frac{k^2+2\alpha k}{(k+\alpha)^2}.
\end{align*}
On the other hand, the greedy algorithm may select first to cover the $k$ one-columns using the 1-value entries $(k+1,j)$ for $j\leq k$ and, then, bundle the zero-column together with only one 1-value entry in some of the first $k$ rows.
The partition value of the greedy algorithm is then
\begin{align*}
\mbox{GREEDY} &=\frac{k+(k+1)\alpha}{(k+\alpha)(\alpha+1)}.
\end{align*}
Hence, the ratio between the two partition values is 
\begin{align*}
\frac{\mbox{GREEDY}}{\mbox{OPT}} &\leq \frac{(k+\alpha)(k+(k+1)\alpha)}{(k^2+2\alpha k)(\alpha+1)}.
\end{align*}
Pick an arbitrarily small $\delta>0$; then, there exist a value for $\alpha$ (significantly higher than $k$) so that the above ratio satisfies $\frac{\mbox{GREEDY}}{\mbox{OPT}} \leq \frac{k+1}{2k}+\delta$. The theorem follows by picking $k$ sufficiently large and $\delta$ sufficiently small.
\end{proof}

We can use the more sophisticated smooth greedy algorithm of Vondr\'ak \cite{V08}, which uses value queries to obtain the following.

\begin{corollary}
There exists a $(1-1/e)$-approximation algorithm for the asymmetric binary matrix partition problem.
\end{corollary}

One might hope that due to the particular form of functions $R_i$, better approximation guarantees might be possible using the $(1-1/e+\epsilon)$-approximation algorithm of Feige and Vondr\'ak \cite{FV10} which requires that demand queries of the form
\begin{quote}
given a price $q_j$ for every item $j \in [m]$, select the bundle $S$ that maximizes the difference $R_i(S)- \sum_{j \in S}{q_j}$
\end{quote}
can be answered in polynomial time. Unfortunately, in our setting, this is not the case in spite of the very specific form of the function $R_i$.

\begin{lemma}\label{lem:demand-hard}
Answering demand queries associated with the asymmetric binary matrix partition problem are NP-hard.
\end{lemma}

\begin{proof}
We use reduction from {\sc Partition} to show that the following (very restricted) decision version DQ of a demand query is NP-hard.
\begin{quote}
DQ: Given a $1\times m$ binary matrix $A$, probabilities $p_j$ and prices $q_j$ for column $j \in [m]$, is there a set $S\subseteq [m]$ such that $R_i(S)- \sum_{j \in S}{q_j}\geq 5/18$?
\end{quote}

We start from an instance of {\sc Partition} consisting of a collection $C$ of $t$ items of integer size $w_1$, $w_2$, ..., $w_t$ and the question of whether there exists a subset $Y\subseteq C$ of items such that
$$\sum_{j\in Y}{w_j} = \sum_{j\in C\setminus Y}{w_j} = \frac{1}{2}\sum_{j\in C}{w_j}.$$
Define $W=\sum_{j\in C}{w_j}$. Given this instance, we construct an instance of DQ with $m=t+1$ as follows. The binary matrix $A$ consists of a single row that contains $t$ 1-value entries with associated probabilities $\frac{w_1}{2W}$, $\frac{w_2}{2W}$, ..., $\frac{w_t}{2W}$ and a 0-value entry with associated probability $1/2$. Set the prices as $q_j=\frac{5w_j}{18W}$ for $j=1, ..., t$ and $q_{t+1}=0$.

By the definition of the function $R_i$, given a set $S\subseteq [t+1]$, we have
\begin{align*}
R_i(S) -\sum_{j\in S}{q_j} &= \frac{1}{2W}\sum_{j\in S\setminus \{t+1\}}{w_j} + \frac{\frac{1}{4W}\sum_{j\in [t]\setminus S}{w_j}}{\frac{1}{2}+\frac{1}{2W}\sum_{j\in [t]\setminus S}{w_j}} - \frac{5}{18W}\sum_{j\in S\setminus \{t+1\}}{w_j}\\
&=\frac{2}{9} - \frac{2}{9W}\sum_{j\in [t]\setminus S}{w_j}+\frac{\sum_{j\in [t]\setminus S}{w_j}}{2W+2\sum_{j\in [t]\setminus S}{w_j}}.
\end{align*}
Now, consider the function $f(z)= \frac{2}{9}-\frac{2z}{9W} +\frac{z}{2W+2z}$; the equality above implies that
$$R_i(S) -\sum_{j\in S}{q_j} = f\left(\sum_{j\in [t]\setminus S}{w_j}\right).$$
By nullifying the derivative of function $f$, we obtain that is has a unique maximum at $z=W/2$. Since $f(W/2)=5/18$, the instance of DQ is equivalent to asking whether there exists a set $S$ such that $\sum_{j\in [t]\setminus S}{w_j}=W/2$, which is equivalent to asking whether there exists a set of items of total size $W/2$ in the instance of {\sc Partition}.
\end{proof}


\begin{thebibliography}{99}

\bibitem{A70}
G. A. Akerlof. The market for `lemons': Quality uncertainty and the market mechanism. {\em Quaterly Journal of Economics}, 84, pages 488--500, 1970.

\bibitem{AFGT13}
N. Alon, M. Feldman, I. Gamzu, and M. Tennenholtz. The asymmetric matrix partition problem. In {\em Proceedings of the 9th Conference on Web and Internet Economics (WINE)}, LNCS, pages 1--14 2013.

\bibitem{ACK09}
S. Athanassopoulos, I. Caragiannis, and C. Kaklamanis. Analysis of approximation algorithms for k-set cover using factor-revealing linear programs. {\em Theory of Computing Systems}, 45(3), pages 555--576, 2009.

\bibitem{ACKK09}
S. Athanassopoulos, I. Caragiannis, C. Kaklamanis, and M. Kyropoulou. An improved approximation bound for spanning star forest and color saving. In {\em Proceedings of the 34th International Symposium on Mathematical Foundations of Computer Science (MFCS)}, pages 90--101, 2009.

\bibitem{C09}
I. Caragiannis. Wavelength management in WDM rings to maximize the number of connections. {\em SIAM Journal on Discrete Mathematics}, 23(2), pages 959--978, 2009.

\bibitem{CKK13}
I. Caragiannis, C. Kaklamanis, and M. Kyropoulou. Tight approximation bounds for combinatorial frugal coverage algorithms. {\em Journal of Combinatorial Optimization}, 26(2), pages 292--309, 2013.

\bibitem{CS82}
V. Crawford and J. Sobel. Strategic information transmission. {\em Econometrica}, 50, pages 1431--1451, 1982.

\bibitem{CM85}
J. Cremer and R. P. McLean. Optimal selling strategies under uncertainty for a discriminating monopolist when demands are interdependent. {\em Econometrica}, 53, pages 345--361, 1985.

\bibitem{CM88}
J. Cremer and R. P. McLean. Full extraction of the surplus in bayesian and dominant strategy auctions. {\em Econometrica}, 56, pages 1247--1257, 1988.

\bibitem{EFGPT12}
Y. Emek, M. Feldman, I. Gamzu, R. Paes Leme, and M. Tennenholtz. Signaling schemes for partition value maximization. In {\em Proceedings of the 13th ACM Conference on Electronic Commerce (EC)}, pages 514--531, 2012.

\bibitem{FV10}
U. Feige and J. Vondr\'ak. The submodular welfare problem with demand queries. {\em Theory of Computing}, 6, pages 247--290, 2010.

\bibitem{GNS07}
A. Ghosh, H. Nazerzadeh, and M. Sundararajan. Computing optimal bundles for sponsored search. In {\em Proceedings of the 3rd International Workshop on Web and Internet Economics (WINE)}, pages 576--583, 2007.

\bibitem{KLMM08}
S. Khot, R. Lipton, E. Markakis, and A. Mehta. Inapproximability results for combinatorial auctions with submodular utility functions. {\em Algorithmica}, 52, pages 3--18, 2008.

\bibitem{LLN06}
B. Lehmann, D. J. Lehmann, and N. Nisan. Combinatorial auctions with decreasing marginal utilities. {\em Games and Economic Behavior}, pages 270--296, 2006.

\bibitem{LM10}
J. Levin and P. Milgrom. Online advertising: Heterogeneity and conflation in market design. {\em American Economic Review}, 100, pages 603--607, 2010.

\bibitem{M10}
P. Milgrom. Simplified mechanisms with an application to sponsored-search auctions. {\em Games and Economic Behavior}, 70, pages 62--70, 2010.

\bibitem{MW82a}
P. R. Milgrom and R. J. Weber. A theory of auctions and competitive bidding. {\em Econometrica}, 50, pages 1089--1122, 1982.

\bibitem{MW82b}
P. R. Milgrom and R. J. Weber. The value of information in a sealed-bid auction. {\em Journal of Mathematical Economics}, 10, pages 105--114, 1982.

\bibitem{MS12}
P. B. Miltersen and O. Sheffet. Send mixed signals - Earn more, work less. In {\em Proceedings of the 13th ACM Conference on Electronic Commerce (EC)}, pages 234--247, 2012.

\bibitem{V08}
J. Vondr\'ak. Optimal approximation for the submodular welfare problem in the value oracle model. In {\em Proceedings of the 40th ACM Symposium on Theory of Computing (STOC)}, pages 67-74, 2008.

\end{thebibliography}
\end{document}